\documentclass[3p,times]{elsarticle}

\usepackage{ecrc}

\volume{00}

\firstpage{1}

\journalname{Theoretical computer science}

\runauth{Chen Wang et al.}

\jid{procs}

\jnltitlelogo{TCS}

\CopyrightLine{2023}{Published by Elsevier Ltd.}

\usepackage{amssymb,color}

\usepackage[figuresright]{rotating}

\usepackage{amsmath}
\usepackage{amsthm}
\usepackage{graphics}
\usepackage{subfigure}
\usepackage{lineno,hyperref}
\usepackage{cases}
\usepackage{bm}
\usepackage[linesnumbered,ruled,vlined]{algorithm2e}
\usepackage{hyperref}
\usepackage[noend]{algpseudocode}
\modulolinenumbers[1]

\newtheorem{proposition}{Proposition}[section]

\newtheoremstyle{prestyle}
{0} 
{\topsep} 
{\itshape} 
{} 
{\bfseries} 
{.} 
{.5em} 
{} 
\theoremstyle{prestyle}

\theoremstyle{definition}

\begin{document}

\begin{frontmatter}



\dochead{}

\title{Approximation Algorithm of Minimum All-Ones Problem for Arbitrary Graphs}

\author[1]{Chen Wang}
\ead{2120220677@mail.nankai.edu.cn}

\author[2]{Chao Wang\corref{cor1}}
\ead{wangchao@nankai.edu.cn}

\author[3]{Gregory Z. Gutin}
\ead{gutin@cs.rhul.ac.uk}

\author[4]{Xiaoyan Zhang}
\ead{xiaoyanice@aliyun.com}

\address[1]{Address: College of Software, Nankai University}
\address[2]{Address: College of Software, Nankai University}
\address[3]{Address: Royal Holloway, University of London}
\address[4]{School of Mathematical Science, Institute of Mathematics, Nanjing Normal University, Jiangsu 210023, China}

\cortext[cor1]{Corresponding author.}



\begin{abstract}
	Let $G=(V, E)$ be a graph and let each vertex of $G$ has a lamp and a button. Each button can be of $\sigma^+$-type or $\sigma$-type.
	Assume that initially some lamps are on and others are off. The button on vertex $x$ is of $\sigma^+$-type ($\sigma$-type, respectively) if pressing the button
	changes the lamp states on $x$ and on its neighbors in $G$ (the lamp states on the neighbors of $x$ only, respectively). Assume that there is a set $X\subseteq V$ such that pressing buttons on vertices of $X$ lights all lamps on vertices of $G$. In particular, it is known to hold when initially all
	lamps are off and all buttons are of $\sigma^+$-type. 
	
	Finding such a set $X$ of the smallest size is NP-hard even if initially all
	lamps are off and all buttons are of $\sigma^+$-type. Using a linear algebraic approach we design a polynomial-time approximation algorithm for the problem such that for the set $X$ constructed by the algorithm, we have $|X|\le \min\{r,(|V|+{\rm opt})/2\},$ where $r$ is the rank of a (modified) adjacent matrix of $G$ and ${\rm opt}$ is the size of an optimal solution to the problem. 
	To the best of our knowledge, this is the first polynomial-time approximation algorithm for the problem with a nontrivial approximation guarantee. 
\end{abstract}

\begin{keyword}
	
	
	lamp lighting problem
	\sep minimum all-ones problem
	\sep approximation algorithm
	\sep complexity
	
\end{keyword}

\end{frontmatter}


\section{Introduction \label{s1}}

The all-ones problem is a fundamental problem in applied mathematics, first proposed by Sutner in 1988 \cite{Sutner1988Problem}. This problem has applications in linear cellular automata, as discussed in \cite{Sutner1989} and the references therein. To illustrate the problem, consider an $n\times n$ grid with each area having a light lamp and a switch, and every lamp is initially off. Turning the switch on in some area lights the lamp in the area and the lamps in neighboring areas. 
Is there a set $X$ of areas such that turning the switches on in $X$ will turn on all the lamps? This problem can be extended to all graphs and we will call it the {\em all-ones problem}. 
Sutner \cite{Sutner1989} proved that a solution $X$ exists for every graph. 
Later, several simple proofs of this result were given or rediscovered \cite{Caro96,Cowen99,Eriksson2001,Goldwasser97,Lossers1993}.

Many variants of the all-ones problem have been introduced and studied \cite{Barua96,broersma2007complexity,DodisW01,Eriksson2001,li2008general,li2005new,Sutner1990} over years. There are two important generalizations of the all-ones problem: (i) the initial state of lamps and switches can be arbitrary, i.e., some are on and the others are off, and (ii) every switch can be either of $\sigma^+$-type which changes the states of the lamp on its vertex and the lamps on the neighbors of its vertex or $\sigma$-type which changes the states of the lamps on the neighbors of its vertex only. As a result of these two generalizations, the {\em generalized all-ones problem} may not have a solution $X$ which lights all lamps. This generalized problem is studied in this paper. Under the condition that such a solution $X$ exists for the generalized all-ones problem, it is natural to ask for $X$ of minimum size. 

Unfortunately, this minimization problem is NP-hard even for all-ones problem \cite{Sutner1988}; we will call the minimization all-ones problem the {\em min all-ones problem}. Galvin and Lu both proved that the min all-ones problem of trees can be solved in linear time \cite{Galvin1989,lu2007minimum}. Building on this, Chen proposed an algorithm for solving the min generalized all-ones problem on trees, with linear complexity \cite{Chen2004}.  Manuel et al. provided solutions for some of the widely studied architectures, such as binomial trees, butterfly, and benes networks \cite{manuel2012all}. Fleischer and Yu provided a detailed survey of the generalized all-ones problem \cite{fleischer2013survey}. More recently, Zhang extended the all-ones problem to the all-colors problem, in which each lamp had other states besides being on and off, and obtained additional findings on the all-colors problem \cite{zhang2019solutions}.

Although significant research has been conducted on the all-ones problem on special graphs, such as trees, resulting in efficient algorithms, no polynomial-time approximation algorithms have been designed for the min all-ones problem on general graphs. Trees and cyclic graphs only represent a fraction of general graphs. In practical engineering scenarios, complex graphs are more common. In this paper, we design a polynomial-time approximation algorithm for the min generalized all-ones problem. If the problem has a solution, our algorithm outputs a solution $X$ such that $|X|\le \min\{r,(|V|+{\rm opt})/2\},$ where the rank of a (modified) adjacent matrix of $G$ and ${\rm opt}$ is the size of an optimal solution to the problem.


Apart from the introduction, this paper contains three sections. In Section \ref{s2}, we introduce our approximation algorithm in detail. Section \ref{s3} shows the theoretical analysis and performance evaluation of this algorithm. Section \ref{s4} summarizes all the work of this paper and discusses future work.

\section{Approximation algorithm of min generalized all-ones problem\label{s2}}
\subsection{Linear algebraic formulation of min generalized all-ones problem\label{s21}}

It is not hard to see that the min generalized all-ones problem can be described as the following linear integer program over  $\mathbb{F}_2$. 
For an arbitrary graph $G=(V,E)$ with $V=\{v_1,\dots ,v_n\}$ we can get its modified adjacency matrix $A = (a_{ij})_{n\times n}$ such that for all $i\ne j$, $a_{ij}=1$ if $v_iv_j\in E$ and $a_{ij}=0$ otherwise, 
and for all $i\in \{1,2,\dots ,n\},$ $a_{ii}=1$ ($a_{ii}=0$, respectively) if  the switch on $v_i$ is of $\sigma^+$-type (of $\sigma$-type, respectively).
Combined with the initial state $B = (b_1,b_2,\cdots,b_n)$, where $b_i=0$ if the lamp on vertex $v_i$ is initially on and $b_i=1$ if the lamp is initially off, we can construct a system of linear equations $AU=B$ over $\mathbb{F}_2$. The solution to this problem is the minimum of $\sum U = \sum_{i=1}^{n}u_i$. 

Suppose the rank of $A$ is $r$ and the corank is $m$ so that $m+r=n$. If $a_{ii} = 1$ for all $i \in\{1,2,\cdots,n\} $, the system of equations $AU=B$ must have a solution, but if $a_{ii} = 0$, the system may not necessarily have a solution. However, as long as the system has at least one solution $\gammaup = (\gamma_1,\gamma_2,\cdots,\gamma_n)^T$, we can find all solutions of the system using the following system combining $\gammaup$ with 
the fundamental solution set $\etaup = (\eta_1,\eta_2,\cdots,\eta_m)$ within time $O(n^3)$. Here $x_i$ is the coefficient of the column vector $\eta_i=(\eta_{1i},\dots , \eta_{ni})^T$.
\begin{equation}
\label{e1}
\begin{split}
	\etaup X + \gammaup 
	& =
	\begin{Bmatrix}
		\eta_{11} & \eta_{12} & \cdots & \eta_{1m}\\
		\eta_{21} & \eta_{22} & \cdots & \eta_{2m}\\
		\eta_{31} & \eta_{32} & \cdots & \eta_{3m}\\
		\vdots & \vdots  & \ddots & \vdots\\
		\eta_{n1} & \eta_{n2} & \cdots & \eta_{nm}\\
	\end{Bmatrix}
	\begin{Bmatrix}
		x_{1} \\ x_{2} \\ \vdots \\x_{m}
	\end{Bmatrix}
	+
	\begin{Bmatrix}
		\gamma_{1} \\ \gamma_{2}\\ \gamma_{3} \\ \vdots \\\gamma_{n}
	\end{Bmatrix}
\end{split}
\end{equation}

The problem is how to find the appropriate column vector $X$ to minimize $\sum U$, under the condition that $X$ has a total of $2 ^ m$ values. This problem was proven to be an NP-complete \cite{Sutner1988}. 
Therefore, the next subsection provides an approximation algorithm running in polynomial time.

\subsection{Approximation algorithm\label{s22}}
Firstly, it can be observed that the polynomial time complexity (not exceeding $O(n^3)$)  of finding the matrix $(\eta_1,\eta_2,\cdots,\eta_m)$ and the special solution $\gammaup$ makes this process cost-effective in solving NP-complete problems. Secondly, it is challenging to identify alternative methods capable of directly computing the optimal solution without obtaining all the solutions. Even if such a solution is obtained, verification is often infeasible. 
When $\etaup$ and $\gammaup$ are known, we need to find the $X$ that minimizes $\sum U$.
\begin{equation}
\label{e2}
\begin{split}
	\etaup X + \gammaup 
	& =
	\begin{Bmatrix}
		\eta_{11} & \eta_{12} & \cdots & \eta_{1m}\\
		\eta_{21} & \eta_{22} & \cdots & \eta_{2m}\\
		\eta_{31} & \eta_{32} & \cdots & \eta_{3m}\\
		\vdots & \vdots  & \ddots & \vdots\\
		\eta_{n1} & \eta_{n2} & \cdots & \eta_{nm}\\
	\end{Bmatrix}
	\begin{Bmatrix}
		x_{1} \\ x_{2} \\ \vdots \\x_{m}
	\end{Bmatrix}
	+
	\begin{Bmatrix}
		\gamma_{1} \\ \gamma_{2}\\ \gamma_{3} \\ \vdots \\\gamma_{n}
	\end{Bmatrix} \\
	& =
	\begin{Bmatrix}
		\delta_{1} \\ \delta_{2}\\ \delta_{3} \\ \vdots \\\delta_{n}
	\end{Bmatrix}
	+
	\begin{Bmatrix}
		\gamma_{1} \\ \gamma_{2}\\ \gamma_{3} \\ \vdots \\\gamma_{n}
	\end{Bmatrix} = 
	\begin{Bmatrix}
		u_{1} \\ u_{2}\\ u_{3} \\ \vdots \\ u_{n}
	\end{Bmatrix}
\end{split}
\end{equation}

\begin{proposition}
Row exchanges of matrix $\etaup$ do not change $\sum U$.
\end{proposition}
\begin{proof}
Multiply both sides of Equation \ref{e2} by matrix $P$ as shown in Equation \ref{e22}. Matrix $P$ is a product of elementary matrices that perform row exchanges. This operation essentially reorders the elements of vector $U$, but does not change the $\sum U$.
\begin{equation}
	\label{e22}
	P (\etaup X + \gammaup) = P (\deltaup + \gammaup) = P U
\end{equation}
\end{proof}

\begin{proposition}
Column transformation of matrix $\etaup$ does not change $\sum U$.
\end{proposition}

\begin{proof}
Let $Q_{m*m}$ be a full rank matrix, and $QZ=X$, with the following equation. 
\begin{equation}
	\label{e3}
	\begin{split}
		\etaup X + \gammaup 
		& =
		\etaup QZ + \gammaup  \\
		& =
		(\etaup Q)Z + \gammaup \\
		& =
		\epsilonup Z + \gammaup \\
		& =
		\deltaup + \gammaup \\
		& = U 
	\end{split}
\end{equation}
$Q$ is the transition matrix between $X$ and $Z$, and $Q$ is full rank. When we find that $X$ makes $\sum U$ the smallest, we can definitely find the corresponding $Z$, so that the obtained $U$ is the same.
\end{proof}

We can transform the $\etaup$ column into an echelon form using row exchanges and column transformations, as shown in the following equation, with a complexity of $O (m^2n)$. The question mark indicates that the value of the number is uncertain, which may be 0 or 1. We can divide the matrix into $m+1$ parts based on the echelon and assume the last line of the $i$-th part is line $k_i$ ($i = 0,1,\cdots,m$) for the rank of matrix $\etaup$ is always $m$. Part 0 is the most special, with all elements in each row being 0. To ensure that Equation \ref{e3} holds, there should be $(u1,u2,\cdots,u_{k_0})=(\gamma_1,\gamma_2,\cdots,\gamma_{k_0})$.
\begin{equation}
\label{e4}
\etaup Q = \epsilonup = 
\begin{Bmatrix}
	0 & 0 & 0 & 0 & \cdots & 0\\
	\vdots & \vdots & \vdots & \vdots & \ddots & \vdots\\
	1 & 0 & 0 & 0 & \cdots & 0\\
	1 & 0 & 0 & 0 & \cdots & 0\\
	\vdots & \vdots & \vdots & \vdots & \ddots & \vdots\\
	\epsilon_{(k_1+1)1} & 1 & 0 & 0 & \cdots & 0\\
	\epsilon_{(k_1+2))1} & 1 & 0 & 0 & \cdots & 0\\
	\vdots & \vdots & \vdots & \vdots & \ddots & \vdots\\
	\epsilon_{(k_2+1)1} & \epsilon_{(k_2+1)2} & 1 & 0 & \cdots & 0\\
	\epsilon_{(k_2+2)1} & \epsilon_{(k_2+2)2} & 1 & 0 & \cdots & 0\\
	\vdots & \vdots & \vdots & \vdots & \ddots & \vdots\\
	\epsilon_{k_m1} & \epsilon_{k_m2} & \epsilon_{k_m3} & \epsilon_{k_m4} & \cdots & 1\\
\end{Bmatrix}
\end{equation}
In the following $m$ parts, we will use greedy algorithms to solve for the $Z$ value on the Echelon of each part. Part 1 of the linear Equation \ref{e4} is shown in Equation \ref{e5}. $(\gamma_{k_0+1},\gamma_{k_0+2},\cdots,\gamma_{k_1})$ is known and $(\delta_{k_0+1},\delta_{k_0+2},\cdots,\delta_{k_1})$ is unknown. It is important to ensure that $\delta_i$ is as similar to $\gamma_i$ as possible. At this moment $z_1$ only has two possible values: 0 and 1. Therefore, the idea of a greedy algorithm is adopted here. If there are more 0's than 1's in the range from $\gamma_{k_0+1}$ to $\gamma_{k_1}$, then $z_1$ is set to 0. If there are more 1's than 0's, then $z_1$ is set to 1. Therefore, we can directly obtain the value of $x_1$ by solving it here, while ensuring that $\sum_{i=k_0+1}^{k_1}u_i  \leq (k_1 - k_0) / 2$.
\begin{equation}
\label{e5}
\begin{cases}
	&  z_1= \delta_{k_0+1}\\
	&  z_1= \delta_{k_0+2}\\
	&  z_1= \delta_{k_0+3}\\
	&  \ \ \ \ \ \ \vdots  \\
	&  z_1= \delta_{k_1}
\end{cases}
\ \ \ \ \ \text{Compare to} \ \ \ \ 
\begin{matrix}
	\gamma_{k_0+1} \\
	\gamma_{k_0+2} \\
	\gamma_{k_0+3} \\
	\vdots  \\
	\gamma_{k_1}
\end{matrix}
\end{equation}

The value of $z_2$ can be calculated through $z_1$. Part 2 of Equation \ref{e4} can be written as shown in Equation \ref{e6}. $(\gamma_{k_1+1},\gamma_{k_1+2},\cdots,\gamma_{k_2})$ is known, and $(\delta_{k_1+1},\delta_{k_1+2},\cdots,\delta_{k_2})$ needs to satisfy the Equation \ref{e4} and be as similar to $(\gamma_{k_1+1},\gamma_{k_1+2},\cdots,\gamma_{k_2})$ as possible. The variables in Equation \ref{e6} are $z_1$ and $z_2$, and $z_1$ has been solved before through a greedy algorithm, so the unknown variable is only $z_2$. Since $\epsilon_{i1}z_1$ are constants, we can move them from the left side of the equation to the right side, and these two equation systems are obviously equivalent. Then, we need to ensure that $\epsilon_{(k_1+i)1}z_1+\delta_{k_1+i}$ is as similar to $\gamma_{k_1+i}$ as possible. It can be seen that another transformation can be carried out, which is equivalent to making $\delta_{k_1+i}$ as similar to $\epsilon_{(k_1+i)1}z_1+\gamma_{k_1+i}$ as possible. In this way, we have separated the variables: the left side of the equation is the variable $z_2$, the right side of the equation is the variable $\delta_{k_1+i} (\delta_{k_1+i} = z_2)$, and the column of $\epsilon_{(k_1+i)1}z_1+\gamma_{k_1+i}$ are constants. At this point, we find that part 2 of Equation \ref{e4} has been transformed to be very similar to part 1. Therefore, if there are more 0's than 1's in the range from $\epsilon_{(k_1+i)1}z_1+\gamma_{k_1+1}$ to $\epsilon_{(k_2)1}z_1+\gamma_{k_2}$, then $z_2$ is set to 0. If there are more 1's than 0's, then $z_2$ is set to 1. Therefore, the value of $z_2$ can be solved here and $\sum_{i=k_1+1}^{k_2}u_i  \leq (k_2 - k_1) / 2$ is ensured.

After obtaining the value of $z_2$, the value of $\epsilon_{i1}z_1 + \epsilon_{i2}z_2$ can be calculated, and the value of $z_3$ can be calculated again. Following this pattern, the values of $Z = (z_1, z_2, \cdots, z_m)$ can be obtained. Then $\epsilonup Z + \gammaup = U$, we obtain $U$. The complete algorithm is shown in Algorithm \hyperref[a1]{1}.

\begin{equation}
\label{e6}
\begin{split}
	\begin{cases}
		&  \epsilon_{(k_1+1)1}z_1+z_2= \delta_{k_1+1}\\
		&  \epsilon_{(k_1+2)1}z_1+z_2= \delta_{k_1+2}\\
		&  \epsilon_{(k_1+3)1}z_1+z_2= \delta_{k_1+3}\\
		&  \ \ \ \ \ \ \vdots  \\
		&  \epsilon_{(k_2)1}z_1+z_2= \delta_{k_2} 
	\end{cases}
	\ \ \ \ \ \ \ \ \ \ &\text{Compare to} \ \ \ \ 
	\begin{matrix}
		\gamma_{k_1+1} \\
		\gamma_{k_1+1} \\
		\gamma_{k_1+1} \\
		\vdots  \\
		\gamma_{k_2}
	\end{matrix}
	\\\\\Downarrow\\\\
	\begin{cases}
		&  z_2= \delta_{k_1+1}+\epsilon_{(k_1+1)1}z_1\\
		&  z_2= \delta_{k_1+2}+\epsilon_{(k_1+2)1}z_1\\
		&  z_2= \delta_{k_1+3}+\epsilon_{(k_1+3)1}z_1\\
		&  \ \ \ \ \ \ \vdots  \\
		&  z_2= \delta_{k_2}+ \epsilon_{(k_2)1}z_1
	\end{cases}
	\ \ \ \ \ \ \ \ \ \ &\text{Compare to} \ \ \ \ 
	\begin{matrix}
		\gamma_{k_1+1} \\
		\gamma_{k_1+1} \\
		\gamma_{k_1+1} \\
		\vdots  \\
		\gamma_{k_2}
	\end{matrix}
	\\\\\Downarrow\\\\
	\begin{cases}
		&  z_2= \delta_{k_1+1}\\
		&  z_2= \delta_{k_1+2}\\
		&  z_2= \delta_{k_1+3}\\
		&  \ \ \ \ \ \ \vdots  \\
		&  z_2= \delta_{k_2}
	\end{cases}
	\ \ \ \ \ \ \ \ \ \text{Compare to}&  \ \ \ \ 
	\begin{matrix}
		\gamma_{k_1+1}+ \epsilon_{(k_1+1)1}z_1 \\
		\gamma_{k_1+2}+ \epsilon_{(k_1+2)1}z_1 \\
		\gamma_{k_1+3}+ \epsilon_{(k_1+3)1}z_1 \\
		\vdots  \\
		\gamma_{k_2}+ \epsilon_{(k_2)1}z_1
	\end{matrix}
\end{split}
\end{equation}

\begin{algorithm}[h]
\label{a1}
\caption{Approximation Algorithm of Minimum All-Ones Problem}
\SetAlgoLined
\KwData{An adjacency matrix $A_{n*n}$ and a initial state $B_{1*n}$}
\KwResult{Answer $U$}
$(\etaup,\gammaup,m)$ = solveEquations($A$,$B$)\;  

\If{$m == 0$ and $\gammaup$ is null}{
	return null\;
}

\If{$m == 0$ and $\gammaup$ is not null}{
	return $\gammaup$\;
}

$(P,\epsilonup,Q)$ = matrixEchelon($\etaup, \gammaup$)\;

$K$ = calculatePart($\epsilonup$)\;

\For{$i$ from 1 to $m$}{
	$cnt,tmp = 0$\;
	
	\For{$j$ from $K[i-1]+1$ to $K[i]$}{
		\For{$p$ from 1 to $i-1$}{
			$tmp = tmp \oplus (\epsilonup[j][p]*X[p])$
		}
		$cnt = cnt + (tmp \oplus \gammaup[j])$;
	}
	\If{$cnt \leq K[i]-K[i-1])/2$}{
		$X[i] = 0$\;
	}\Else{
		$X[i] = 1$\;
	}
}
$U = P * (\epsilonup * X + \gammaup$)\;

return $U$\;
\end{algorithm}

\section{Algorithm performance evaluation\label{s3}}	
In this section, we present the complexity of Algorithm \hyperref[a1]{1} and analyze its approximation guarantees.

\begin{proposition}
Algorithm \hyperref[a1]{1} has a complexity of $O(n^3)$, and if the fundamental solution set $\etaup$ for the equation $AU=B$ has been obtained and is in column echelon form, then the complexity will reduce to $O(mn)$.
\end{proposition}

\begin{proof}
In Algorithm \hyperref[a1]{1}, step 1 involves solving a system of linear equations, which has a complexity of $O(n^3)$. Step 8 involves transforming the matrix $etaup$ into column echelon form, which has a complexity of $O(m^2n)$ where $m \leq n$. Step 2 to 7 is $O(1)$. Step 9 involves calculating the location of pivots in the column echelon matrix $\epsilonup$, which has a complexity of $O(mn)$. Steps 10 to 24 consist of a nested loop with three layers. However, each element in the matrix $\epsilonup$ is only accessed once, resulting in a complexity of $O(mn)$.
\end{proof}

\begin{proposition} If a given instance $\cal I$ of the min generalized all-ones problem has a solution, 
the value ${\rm sol}$ of the solution  obtained by Algorithm \hyperref[a1]{1} satisfies ${\rm sol}\le r$, where $r$ is the rank of the matrix $A$.
\end{proposition}

\begin{proof}
In Equation \ref{e5} and \ref{e6}, if $\delta_i=\gamma_i$, then the resulting $u_i$ will be 0. In each part, we always make more $u_i$ equal to 0, so each part has at least one $u_i$  that takes on the value of 0. Furthermore, the rank of $\etaup$ is $m=n-r$ because $\etaup$ is the fundamental solution set of the system $AU=B$. Therefore, at least $m$ values of $u_i$ are 0, so $\sum U \leq n-m=r$.
\end{proof}

\begin{proposition}
If a given instance $\cal I$ of the min generalized all-ones problem has a solution, the value ${\rm sol}$ of the solution obtained by Algorithm \hyperref[a1]{1} satisfies $\rm sol \leq (n+{\rm opt})/2$, where ${\rm opt}$ is the value of an optimal solution of $\cal I$.
\end{proposition}

\begin{proof}
In the Subsection \ref{s22}, we partitioned the matrix $\etaup$ into $m+1$ parts and proved that for the $1$ to $m$ parts, $\sum_{k_i+1}^{k_{i+1}}u \leq (k_{i+1} - k_i) / 2$. Only the 0th part remains to be discussed. The 0th part is quite special in that it contains no variables, only differing in the value of $\gamma$. Let the number of 0's in $\gamma$ in the 0th part be $g_0$ and the number of 1's be $g_1$. $g_0$ indicates that the switch at that point must not be pressed; otherwise, the conditions for the all-ones problem cannot be satisfied. Similarly, $g_1$ indicates that the switch must be pressed. Now we have: 
\begin{equation}
	\label{e7}
	\rm sol \leq g_1 + (n-g_1-g_0)/2 = (n + g_1-g_0)/2
\end{equation}
Then add the parameter ${\rm opt}$. We can easily prove that ${\rm sol}\ge {\rm opt}\ge g_1$, because the switches for these points must be pressed in any case. So we have
\begin{equation}
	\label{e8}
	g_1 \leq {\rm opt} \leq \rm sol \leq (n + g_1-g_0)/2
\end{equation}
Next, we will bound $\rm sol$ by replacing $g_1$ with $\rm opt$ and $g_0$ with 0, resulting in the following expression:
\begin{equation}
	\label{e8}
	{\rm sol} \leq (n+{\rm opt})/2
\end{equation}

\end{proof}

\begin{figure}[h]
\center
\includegraphics[width=0.6\linewidth]{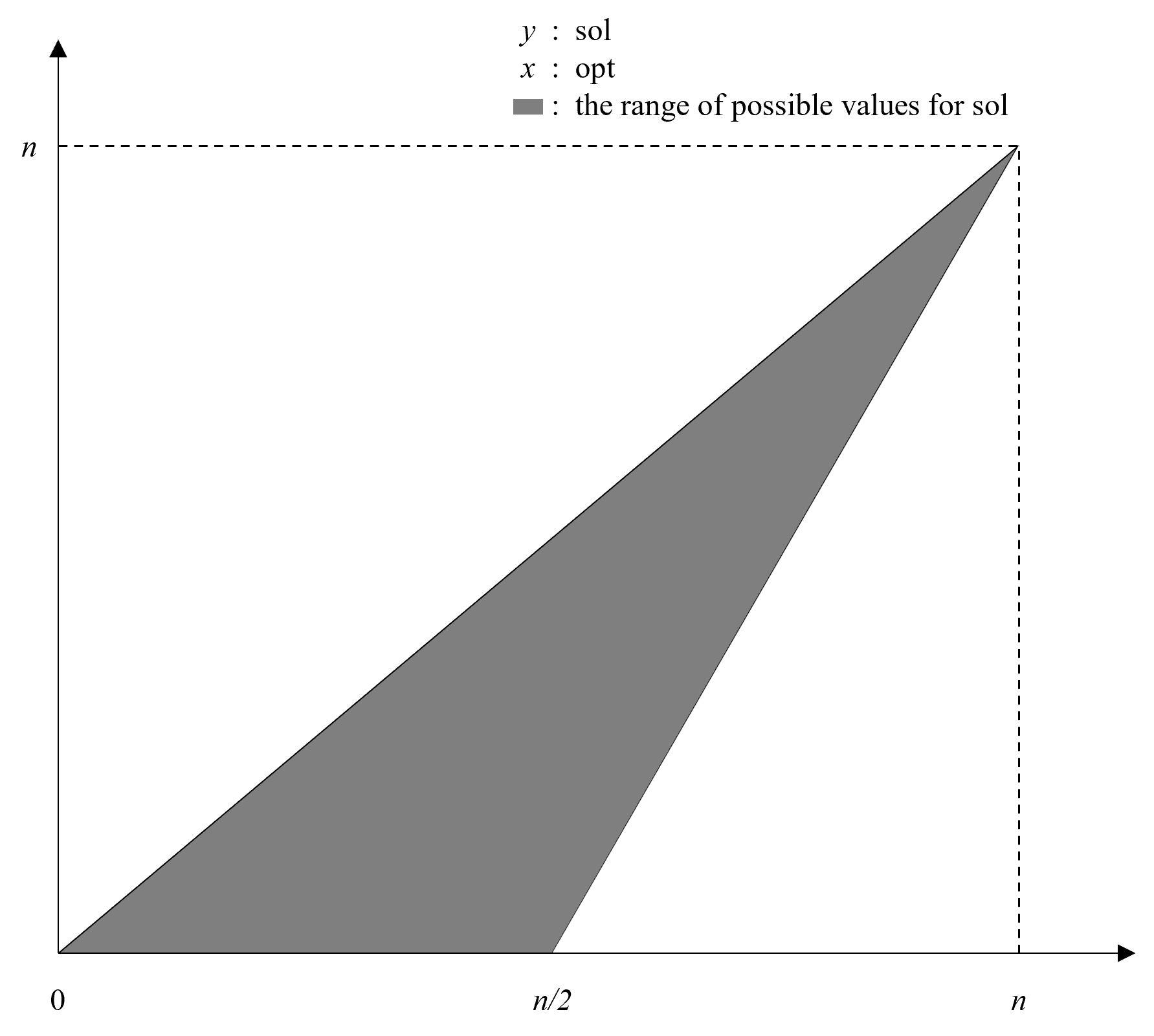}
\caption{The range of possible values for ${\rm sol}$}
\label{f1}
\end{figure}

\section{Conclusion and future work\label{s4}}
This article presents an approximation algorithm for the min generalized all-ones problem on arbitrary graphs, making it possible to process the problem in batches. The algorithm has a complexity of $O(n^3)$. If the equation system $AU=B$ has been solved and the solution is in column echelon form, the complexity will be reduced to $O(n(n-r))$, which is the lowest complexity for general graphs. The upper bound of the solution value ${\rm sol}$ obtained by this algorithm satisfies the inequality ${\rm sol} \leq (n+{\rm opt})/2$ and ${\rm sol} \leq r$. This ensures that the obtained solution, as shown in Figure 1, is always the optimal half.

In future work, there still remain two questions to be solved. One of them is whether there is a polynomial-time algorithm for the min generalized all-ones problem which always finds a solution of size at most $c\dot {\rm opt}$ for some constant $c$?
The other one is whether we can get such an algorithm for the minimum all-colors problem?







\bibliographystyle{elsarticle-harv}
\bibliography{fbref}

\begin{thebibliography}{20}
\expandafter\ifx\csname natexlab\endcsname\relax\def\natexlab#1{#1}\fi
\providecommand{\url}[1]{\texttt{#1}}
\providecommand{\href}[2]{#2}
\providecommand{\path}[1]{#1}
\providecommand{\DOIprefix}{doi:}
\providecommand{\ArXivprefix}{arXiv:}
\providecommand{\URLprefix}{URL: }
\providecommand{\Pubmedprefix}{pmid:}
\providecommand{\doi}[1]{\href{http://dx.doi.org/#1}{\path{#1}}}
\providecommand{\Pubmed}[1]{\href{pmid:#1}{\path{#1}}}
\providecommand{\bibinfo}[2]{#2}
\ifx\xfnm\relax \def\xfnm[#1]{\unskip,\space#1}\fi
\bibitem[{Barua and Ramakrishnan(1996)}]{Barua96}
\bibinfo{author}{Barua, R.}, \bibinfo{author}{Ramakrishnan, S.},
  \bibinfo{year}{1996}.
\newblock \bibinfo{title}{$\sigma$-game, $\sigma^+$-game and two-dimensional
  additive cellular automata}.
\newblock \bibinfo{journal}{Theoretical Computer Science}
  \bibinfo{volume}{154}, \bibinfo{pages}{349--366}.
\bibitem[{Broersma and Li(2007)}]{broersma2007complexity}
\bibinfo{author}{Broersma, H.}, \bibinfo{author}{Li, X.}, \bibinfo{year}{2007}.
\newblock \bibinfo{title}{On the complexity of dominating set problems related
  to the minimum all-ones problem}.
\newblock \bibinfo{journal}{Theoretical Computer Science}
  \bibinfo{volume}{385}, \bibinfo{pages}{60--70}.
\bibitem[{Caro(1996)}]{Caro96}
\bibinfo{author}{Caro, Y.}, \bibinfo{year}{1996}.
\newblock \bibinfo{title}{Simple proffs to three parity theorems}.
\newblock \bibinfo{journal}{Ars Comb.} \bibinfo{volume}{42}.
\bibitem[{Chen et~al.(2004)Chen, Li, Wang and Zhang}]{Chen2004}
\bibinfo{author}{Chen, W.Y.}, \bibinfo{author}{Li, X.}, \bibinfo{author}{Wang,
  C.}, \bibinfo{author}{Zhang, X.}, \bibinfo{year}{2004}.
\newblock \bibinfo{title}{The minimum all-ones problem for trees}.
\newblock \bibinfo{journal}{SIAM Journal on Computing} \bibinfo{volume}{33},
  \bibinfo{pages}{379--392}.
\bibitem[{Cowen et~al.(1999)Cowen, Hechler, Kennedy and Ryba}]{Cowen99}
\bibinfo{author}{Cowen, R.}, \bibinfo{author}{Hechler, S.},
  \bibinfo{author}{Kennedy, J.}, \bibinfo{author}{Ryba, A.},
  \bibinfo{year}{1999}.
\newblock \bibinfo{title}{Inversion and neighborhood inversion in graphs}.
\newblock \bibinfo{journal}{Graph Theory Notes of New York}
  \bibinfo{volume}{37}, \bibinfo{pages}{37--41}.
\bibitem[{Dodis and Winkler(2001)}]{DodisW01}
\bibinfo{author}{Dodis, Y.}, \bibinfo{author}{Winkler, P.},
  \bibinfo{year}{2001}.
\newblock \bibinfo{title}{Universal configurations in light-flipping games},
  in: \bibinfo{editor}{Kosaraju, S.R.} (Ed.), \bibinfo{booktitle}{Proceedings
  of the Twelfth Annual Symposium on Discrete Algorithms, January 7-9, 2001,
  Washington, DC, {USA}}, \bibinfo{publisher}{{ACM/SIAM}}. pp.
  \bibinfo{pages}{926--927}.
\bibitem[{Eriksson et~al.(2001)Eriksson, Eriksson and
  Sj{\"o}strand}]{Eriksson2001}
\bibinfo{author}{Eriksson, H.}, \bibinfo{author}{Eriksson, K.},
  \bibinfo{author}{Sj{\"o}strand, J.}, \bibinfo{year}{2001}.
\newblock \bibinfo{title}{Note on the lamp lighting problem}.
\newblock \bibinfo{journal}{Advances in Applied Mathematics}
  \bibinfo{volume}{27}, \bibinfo{pages}{357--366}.
\bibitem[{Fleischer and Yu(2013)}]{fleischer2013survey}
\bibinfo{author}{Fleischer, R.}, \bibinfo{author}{Yu, J.},
  \bibinfo{year}{2013}.
\newblock \bibinfo{title}{A survey of the game lights out!}
\newblock \bibinfo{journal}{Space-Efficient Data Structures, Streams, and
  Algorithms: Papers in Honor of J. Ian Munro on the Occasion of His 66th
  Birthday} , \bibinfo{pages}{176--198}.
\bibitem[{Galvin(1989)}]{Galvin1989}
\bibinfo{author}{Galvin, F.}, \bibinfo{year}{1989}.
\newblock \bibinfo{title}{Solution to problem 88-8}.
\newblock \bibinfo{journal}{Math. Intelligencer} \bibinfo{volume}{11},
  \bibinfo{pages}{31--32}.
\bibitem[{Goldwasser et~al.(1997)Goldwasser, Klostermeyer and
  Trapp}]{Goldwasser97}
\bibinfo{author}{Goldwasser, J.}, \bibinfo{author}{Klostermeyer, W.},
  \bibinfo{author}{Trapp, G.}, \bibinfo{year}{1997}.
\newblock \bibinfo{title}{Characterizing switch-setting problems}.
\newblock \bibinfo{journal}{Linear and Multilinear Algebra}
  \bibinfo{volume}{43}, \bibinfo{pages}{121--136}.
\bibitem[{Li et~al.(2008)Li, Wang and Zhang}]{li2008general}
\bibinfo{author}{Li, X.}, \bibinfo{author}{Wang, C.}, \bibinfo{author}{Zhang,
  X.}, \bibinfo{year}{2008}.
\newblock \bibinfo{title}{The general $\sigma$ all-ones problem for trees}.
\newblock \bibinfo{journal}{Discrete Applied Mathematics}
  \bibinfo{volume}{156}, \bibinfo{pages}{1790--1801}.
\bibitem[{Li and Zhang(2005)}]{li2005new}
\bibinfo{author}{Li, X.}, \bibinfo{author}{Zhang, X.}, \bibinfo{year}{2005}.
\newblock \bibinfo{title}{New versions of the all-ones problem}.
\newblock \bibinfo{journal}{arXiv preprint math/0512011} .
\bibitem[{Lossers(1993)}]{Lossers1993}
\bibinfo{author}{Lossers, O.}, \bibinfo{year}{1993}.
\newblock \bibinfo{title}{Solution to problem 10197 [1992, 162]-an all-ones
  problem}.
\newblock \bibinfo{journal}{American Mathematical Monthly}
  \bibinfo{volume}{100}, \bibinfo{pages}{806--807}.
\bibitem[{Lu and Li(2007)}]{lu2007minimum}
\bibinfo{author}{Lu, Y.}, \bibinfo{author}{Li, Y.}, \bibinfo{year}{2007}.
\newblock \bibinfo{title}{The minimum all-ones problem for graphs with small
  treewidth}, in: \bibinfo{booktitle}{Combinatorial Optimization and
  Applications: First International Conference, COCOA 2007, Xi’an, China,
  August 14-16, 2007. Proceedings 1}, \bibinfo{organization}{Springer}. pp.
  \bibinfo{pages}{335--342}.
\bibitem[{Manuel et~al.(2012)Manuel, Rajasingh, Rajan and
  Prabha}]{manuel2012all}
\bibinfo{author}{Manuel, P.}, \bibinfo{author}{Rajasingh, I.},
  \bibinfo{author}{Rajan, B.}, \bibinfo{author}{Prabha, R.},
  \bibinfo{year}{2012}.
\newblock \bibinfo{title}{The all-ones problem for binomial trees, butterfly
  and benes networks}.
\newblock \bibinfo{journal}{International Journal of Mathematics and soft
  computing} \bibinfo{volume}{2}, \bibinfo{pages}{1--6}.
\bibitem[{Sutner(1988a)}]{Sutner1988}
\bibinfo{author}{Sutner, K.}, \bibinfo{year}{1988}a.
\newblock \bibinfo{title}{Additive automata on graphs}.
\newblock \bibinfo{journal}{Complex Systems} \bibinfo{volume}{2},
  \bibinfo{pages}{649--661}.
\bibitem[{Sutner(1988b)}]{Sutner1988Problem}
\bibinfo{author}{Sutner, K.}, \bibinfo{year}{1988}b.
\newblock \bibinfo{title}{Problem 88-8}.
\newblock \bibinfo{journal}{The Mathematical Intelligencer}
  \bibinfo{volume}{10}.
\bibitem[{Sutner(1989)}]{Sutner1989}
\bibinfo{author}{Sutner, K.}, \bibinfo{year}{1989}.
\newblock \bibinfo{title}{Linear cellular automata and the {Garden-of-Eden}}.
\newblock \bibinfo{journal}{The Mathematical Intelligencer}
  \bibinfo{volume}{11}, \bibinfo{pages}{49--53}.
\bibitem[{Sutner(1990)}]{Sutner1990}
\bibinfo{author}{Sutner, K.}, \bibinfo{year}{1990}.
\newblock \bibinfo{title}{The $\sigma$-game and cellular automata}.
\newblock \bibinfo{journal}{The American Mathematical Monthly}
  \bibinfo{volume}{97}, \bibinfo{pages}{24--34}.
\bibitem[{Zhang and Wang(2019)}]{zhang2019solutions}
\bibinfo{author}{Zhang, X.}, \bibinfo{author}{Wang, C.}, \bibinfo{year}{2019}.
\newblock \bibinfo{title}{Solutions to all-colors problem on graph cellular
  automata}.
\newblock \bibinfo{journal}{Complexity} \bibinfo{volume}{2019}.
\newblock \bibinfo{note}{Article ID 3164692}.

\end{thebibliography}







\end{document}